\documentclass[sigconf,nonacm]{acmart}
\usepackage{times}  
\usepackage{helvet} 
\usepackage{courier}  
\usepackage{natbib}  
\usepackage{caption} 
\usepackage{microtype}
\usepackage{graphicx}
\usepackage{subfigure}
\usepackage{booktabs} 
\usepackage{multirow}
\graphicspath{{figs/}}
\usepackage{amsmath}
\usepackage{hyperref}
\usepackage{algorithm}
\usepackage{algpseudocode}
\usepackage{tabularx}
\usepackage{amsthm}
\usepackage{amsmath,stackengine}
\usepackage[bb=boondox]{mathalfa}
\usepackage[colorinlistoftodos]{todonotes}
\theoremstyle{definition}

\newtheorem{theorem}{Theorem}[section]

\newtheorem{proposition}[theorem]{Proposition}

\newtheorem{remark}{Remark}[section]

\AtBeginDocument{%
  \providecommand\BibTeX{{%
    \normalfont B\kern-0.5em{\scshape i\kern-0.25em b}\kern-0.8em\TeX}}}

\setcopyright{none}

\begin{document}

\author{Shi Yu}
\authornotemark[1]
\author{Haoran Wang}
\affiliation{%
  \institution{The Vanguard Group}
  \city{Malvern}
  \country{USA}
}
\email{shi_yu,haoran_wang@vanguard.com}

\author{Chaosheng Dong}
\affiliation{%
  \institution{Amazon}
  \city{Seattle}
  \country{USA}}
\email{chaosd@amazon.com}

\title{Learning Risk Preferences from Investment Portfolios Using Inverse Optimization}

\begin{abstract}
The fundamental principle in Modern Portfolio Theory (MPT) is based on the quantification of the portfolio's risk related to performance. Although MPT has made huge impacts to the investment world and prompted the success and prevalence of passive investing, it still has shortcomings in real-world applications. One of the main challenges is that the level of risk an investor can endure, known as \emph{risk-preference}, is a subjective choice that is tightly related to psychology and behavioral science in decision making. This paper presents a novel approach of measuring risk preference from existing portfolios using inverse optimization on mean-variance portfolio allocation framework. Our approach allows the learner to continuously estimate real-time risk preferences using concurrent observed portfolios and market price data. We demonstrate our methods on robotic investment portfolios and real market data that consists of 20 years of asset pricing and 10 years of mutual fund portfolio holdings. Moreover, the quantified risk preference parameters are validated with two well-known risk measurements currently applied in the field. The proposed methods could lead to practical and fruitful innovations in automated/personalized portfolio management, such as Robo-advising, to augment financial advisors' decision intelligence in a long-term investment horizon. 

\end{abstract}


\maketitle

\section{Introduction}

	 In portfolio allocation, one primary goal has been to reconcile empirical information about securities prices with theoretical models of asset pricing under conditions of inter-temporal uncertainty \citep{Cohn75}. The notion of risk preference has been an essential assumption underlying almost all such models \citep{Cohn75}. Its underlying biological, behavioral, and social factors are commonly studied in many other disciplines, i.e., social science \citep{Payne17,guiso2008risk},  behavior science \citep{brennan2011origin}, mathematics \citep{vonneumann1947}, psychology \citep{SokolHessner2009ThinkingLA,Mcgraw10}, and genetics \citep{Linn19}.  For more than half of a century, many measures of risk preference have been developed in various fields, including curvature measures of utility functions \citep{arrow1971theory,Pratt64}, human subject experiments and surveys \citep{rabin2001anomalies,holt2002risk}, portfolio choice for financial investors \citep{guiso2008risk}, labor-supply behavior \citep{chetty2006new}, deductible choices in insurance contracts \citep{cohen2007estimating,szpiro1986relative}, contestant behavior on game shows \citep{post2008deal}, option prices \citep{ait2000nonparametric} and auction behavior \citep{lu2008estimating}. Nowadays, investor-consumers' risk preferences are mainly investigated through one or combination of three ways. The first one is assessing actual behavior, such as inferring households' risk attitudes using regression analysis on historical financial data \citep{Schooley1996RiskAM} and inferring investors' risk preferences from their trading decision using reinforcement learning \citep{HumoudRLRisk}. The second one is assessing responses to hypothetical scenarios about investment choices using online questionnaries (see \citet{barsky1997preference}, \citep{hey1999estimating} and \citep{HumoudRLRisk}. The third one is subjective questions (see \citet{hanna1998theory} for a survey of these different techniques).
	 
	 Despite its profound importance in economics, there remain some limitations with respect to learning the risk preference using classic approaches. First, most practices in place are often insufficient to deal with scenarios where investment decisions are managed by machine learning processes (Robo-advising), and risk preferences are expected as input parameters that can generate new decisions directly (auto-rebalancing). Current Robo-advising process first communicate and categorize clients' risk preferences based on human interpretation, and later map them to the nearest value of a finite set of representative risk preference levels \cite{Capponi2019PersonalizedRE}. This limitation makes existing theories and approaches challenging to deal with prominent situations in which risk preference changes dramatically in an inter-temporal dimension, such as savings, investment, consumption problems, dynamic labor supply decisions, and health decisions \citep{Donoghue18}. Second, real-world risk preference is clearly not as straightforward as many theories have assumed, and perhaps individuals even do not exhibit risk appetite consistently in their behaviors across domains \citep{Donoghue18}. Namely, risk preference is time varying in reality. Nowadays, due to the technological advances, we already have overwhelming behavioral data across all domains, which provide a myriad of additional sources to help us decipher the perplexity of risk preference from different angles. New approaches can be used in conjunction with traditional methods with additional, more nuanced implications that are borne out by data \citep{Donoghue18}. 

	 To tackle aforementioned limitations, we present a novel inverse optimization approach to measure risk preference directly from market signals and portfolios. Our approach is developed based on two fundamental methodologies: convex optimization based Modern Portfolio Theory (MPT) and learning decision-making scheme through inverse optimization. We assume investors are rational and their portfolio decisions are near-optimal. Consequently, their decisions are affected by the risk preference factors through portfolio allocation model. We propose an inverse optimization framework to infer the risk preference factor that must have been in place for those decisions to have been made. Furthermore, we assume risk preference stays constant at the point of decision, but varies across multiple decisions over time, and can be inferred from joint observations of time-series market price and asset allocations. 
	
	 
	 \textbf{Our contributions} We summarize the major contributions of our paper as follows:
	 \begin{itemize}
	     \item To the best of our knowledge, we propose the first inverse optimization approach for learning time varying risk preference parameters of the mean-variance portfolio allocation model, based on a set of observed mutual fund portfolios and underlying asset price data. Furthermore, the flexibility of our approach enables us to move beyond mean-variance and adopt more general risk metrics. 

	     \item We demonstrate our approach in two real-world case studies. First, we generate 20 months of robotic portfolios using an in-house deep reinforcement learning platform. Then, we apply the proposed risk preference algorithm on these robotic portfolios, and the learned risks are consistent with input parameters that drive the robotic investment process.  Next, the proposed method is applied on 20 years of market data and 10 years of mutual fund quarterly holding data. Results show that the proposed method is able to handle learning task that consists of hundreds of assets in portfolio. For portfolios composed by more than thousands of assets, we propose Sector-based and Factor-base projection to improve the computational efficiency. 
	     
	     \item In particular, to the best of our knowledge, it is the first time that inverse optimization approach is formally proposed for portfolio risk learning. It is also the first time robotic portfolios are used to validate the effectiveness of machine learning methods.


	 \end{itemize}
	 
\section{Related Work}

	 The fundamental MPT developed by \citet{markowitz1952portfolio} and its variants assume that investors estimate the risk of the portfolio according to the variability of the expected return. Moreover, \citet{markowitz1952portfolio} assumes that investors make decisions solely based on the preferences of two objectives: the expected return and the risk. The trade-off between the two objectives is typically denoted by a positive coefficient and referred to as \emph{risk tolerance} (or \emph{risk aversion}). Later, \citet{black1992} extends the framework in \citet{markowitz1952portfolio} by blending investors' private expectations, known as Black-Litterman (BL) model. A Bayesian statistical interpretation of BL model is proposed in \citet{He2002TheIB} and an inverse optimization perspective is derived in \citet{bertsimas2012inverse}. Most of these mean-variance based approaches assume an investor's risk preference is known. 
	 
	
	Our work is related to the inverse optimization problem (IOP), in which the learner seeks to
    infer the missing information of the underlying decision model from observed data, assuming that the decision maker is rationally making decision \citep{ahuja2001inverse}.  IOP has been extensively investigated in the operations research community from various aspects \citep{ahuja2001inverse,iyengar2005inverse,Schaefer2009,wang2009cutting,keshavarz2011imputing,chan2014generalized,bertsimas2015data,aswani2016inverse,esfahani2017data,birge2017inverse,barmann2017emulating,chan2018trade,dong2018inferring,dong2018ioponline,chan2019inverse,chan2020inverse,dong2020imopicml}.
    Due to the time varying nature of risk preferences, our work particularly takes the online learning framework in \citet{dong2018ioponline}, which develops an online learning algorithm to infer the utility function or constraints of a decision making problem from sequentially arrived observations of (signal, decision) pairs. This approach makes few assumptions and is generally applicable to learn the underlying decision making problem that has a convex structure. \citet{dong2018ioponline} provide the regret bound and  shows the statistical consistency of the algorithm, which guarantees that the algorithm will asymptotically achieves the best prediction error permitted by the underlying inverse model.
	
	Also related to our work is \citet{bertsimas2012inverse}, in which creates a novel reformulation of the BL framework by using techniques from inverse optimization. There are two main differences between \citet{bertsimas2012inverse} and our paper. First, the problems we study are essentially different.   \citet{bertsimas2012inverse} seek to reformulate the BL model while we focus on learning specifically the investor's risk preferences. Second, \citet{bertsimas2012inverse} consider a deterministic setting in which the parameters of the BL model are fixed and uses only one portfolio as the input. In contrast, we believe that investor's risk preferences are time varying and propose an inverse optimization approach to learn them in an online setting with as many historical portfolios as possible. Such a data-driven approach enables the learner to better capture the time-varying nature of risk preferences and better leverage the power and opportunities offered by `Big data'.
    
    Recently, we note that researchers in reinforcement learning (RL) propose an exploration-exploitation algorithm to learn the investor's risk appetite over time by observing her portfolio choices in different market environments \citep{HumoudRLRisk}. The method is explicitly designed for Robo-advisors. In each period, the Robo-advisor places an investor's capital into one of several pre-constructed portfolios, where each portfolio decision reflects the Robo-advisor's belief of the investor's risk preference. The investor interacts with the Robo-advisor by portfolio selection choices, and such interactions are used to update the Robo-advisor's estimations about the investor's risk profile. The investor's decision function in \citet{HumoudRLRisk} is a simple voting model among different pre-constructed candidate portfolios, and the portfolio allocation process is separate from that decision model. In contrast, we use the portfolio allocation model directly as the decision making model, and IOP allows us to infer risk preference directly from portfolios.

	\section{Problem Setting}

	\subsection{Portfolio Optimization Problem}
	 We consider the Markowitz mean-variance portfolio optimization problem \citep{markowitz1952portfolio}:
	\begin{align}
	\label{mean-variance portfolio}
	\tag*{PO}
    \begin{array}{llll}
         \min\limits_{\mathbf{x} \in \mathbb{R}^{n}} &  \frac{1}{2} \mathbf{x}^{T} Q \mathbf{x} - r \mathbf{c}^{T} \mathbf{x}  \\
    	\;s.t. &  A \mathbf{x} \geq  \mathbf{b},
    \end{array}
	\end{align}
	 where $Q\in\mathbb{R}^{n \times n} \succeq 0 $ is the positive semi-definite covariance matrix of the asset level portfolio returns, $\mathbf{x}$ is portfolio allocation where each element $x_{i}$ is the holding weight of asset $i$ in portfolio, $\mathbf{c} \in \mathbb{R}^{n}$ is a vector of the expected asset level portfolio returns, $r > 0$ is the \emph{risk-tolerance} factor, $A \in \mathbb{R}^{m \times n} (n \leq m)$ is the structured constraint matrix, and $\mathbf{b}\in \mathbb{R}^{m}$ is the corresponding right-hand side in the constraints. 
	 
 In this paper, we have $x_{i} \geq 0$ for each $i \in [n]$  as no shorting is considered. In general, $\mathbf{x}$ represents the portfolio optimized for $n$ assets. In \ref{mean-variance portfolio}, the coefficient $r$ is assigned to the linear term, and thus it represents \emph{risk-tolerance} and larger $r$ indicates more preferable to profit. In \ref{mean-variance portfolio}, if variables $Q, r, \mathbf{c}, A, \mathbf{b}$ are given, the optimal solution $\mathbf{x}^{*}$ can be obtained efficiently via convex optimization. In financial investment, the process is known as finding the optimal portfolio allocations, and we call it the \emph{Forward Problem}.
	 
	
	\subsection{Inverse Optimization through Online Learning}	  
	
	 Now we consider a reverse scenario, assuming we can somehow observe a sequence of optimized portfolios while some parameters in \ref{mean-variance portfolio} such as $r$ are unknown. The problem then becomes learning the hidden decision variables that control the portfolio optimization process. Such type of problem has been systematically investigated in \citet{dong2018ioponline} in an online setting. Formally, consider the family of parameterized decision making problem 
	\begin{align}
	\label{dmp}
	\tag*{DMP}
	\begin{array}{llll}
	     \min\limits_{\mathbf{x} \in \mathbb{R}^{n}} & f(\mathbf{x}, u; \theta) \\
	     \; s.t.  & \mathbf{g}(\mathbf{x}, u; \theta) \leq \mathbf{0}. 
	\end{array}
	\end{align}
	
	 We denote $X(u;\theta)=\{x\in\mathbb{R}^{n} : \mathbf{g}(\mathbf{x}, u; \theta) \leq \mathbf{0}\}$ the feasible region of \ref{dmp}. We let
	\begin{align*}
	    S(u;\theta)=\arg\min\{f(\mathbf{x},u,\theta): \mathbf{x} \in X(u,\theta)\}
	\end{align*}
	 be the optimal solution set of \ref{dmp}.
	 
	Consider a learner who monitors the signal $ u \in \mathcal{U} $ and the decision maker' decision $ \mathbf{x} \in X(u,\theta) $ in response to $ u $. Assume that the learner does not know the decision maker's utility function or constraints in \ref{dmp}. Since the observed decision might carry measurement error or is generated with a bounded rationality of the decision maker, we denote $ \mathbf{y} $ the observed noisy decision for $ u \in \mathcal{U} $. 
	
	In the inverse optimization problem, the learner aims to learn the parameter $ \theta $ of \ref{dmp} from (signal, noisy decision) pairs. In online setting, the (signal, noisy decision) pair, i.e., ($u,\mathbf{y}$), becomes available to the learner one by one. Hence, the learning algorithm produces a sequence of hypotheses  $(\theta_{0},\ldots,\theta_{T})$. Here, $T$ is the total number of rounds, $\theta_{0}$ is an initial hypothesis, and $\theta_{t}$ for $t \geq 2$ is the hypothesis chosen after observing the $t$-th (signal,noisy decision) pair.
	 
	Given a (\emph{signal, noisy decision}) pair ($u,\mathbf{y}$) and a hypothesis $\theta$, the loss function is set as the minimum distance between $\mathbf{y}$ and the optimal solution set S($u;\theta$) in the following form:
	\begin{align}
	\label{loss function}
	l(\mathbf{y},u;\theta) = \min_{\mathbf{x} \in S(u;\theta)}\|\mathbf{y}- \mathbf{x}\|_2^{2}. 
	\end{align}
	
	 Once receiving the $t$-th (signal, noisy decision) pair ($u_{t}$, $\mathbf{y}_{t}$), $\theta$ would be updated by solving the following optimization problem:
	\begin{align}
	\label{update rule}
	\theta_{t}= \arg\min_{\theta \in \Theta} \frac{1}{2} \|\theta -\theta_{t-1} \|_{2}^{2} + \frac{\lambda}{\sqrt{t}} l(\mathbf{y}_{t}, u_{t}; \theta),
	\end{align}
	 where  $\frac{\lambda}{\sqrt{t}}$ is the learning rate in round $t$, and $l(\mathbf{y},u;\theta)$ is the loss function defined in \eqref{loss function}.
	

	\subsection{Learning Time Varying Risk Preferences}
	In the portfolio optimization problem we consider, the (\emph{signal, noisy decision}) pair correspond to (\emph{market price, observed portfolio}). A learner aims to learn the investor’s risk preference from (\emph{price, portfolio}) pairs. More precisely, the goal of the learner is to estimate $r$ in \ref{mean-variance portfolio}. In online setting, the (\emph{price, portfolio}) pairs become available to the learner one by one. Hence, the learning algorithm produces a sequence of hypotheses $(r_{1}, . . . , r_{T+1})$. Here, $T$ is the total number of rounds, $r_{1}$ is an arbitrary initial hypothesis, and $r_{t}$ for $t \geq 1$ is the hypothesis chosen after observing the $(t-1)$-th (\emph{price, portfolio}) pair.
	
	 Given a (\emph{price, portfolio}) pair ($u,\mathbf{y}$) and a hypothesis $r$, similar to \eqref{loss function}, we set the loss function as follows:
	\begin{align}
	\label{loss function for PO}
	l(\mathbf{y},u;r) = \min_{\mathbf{x} \in S(Q,\mathbf{c},A,\mathbf{b};r)}\|\mathbf{y}- \mathbf{x}\|_2^{2}. 
	\end{align}
	where $S(Q,\mathbf{c},A,\mathbf{b};r)$ is the optimal solution set of \ref{mean-variance portfolio}.
	
	 \begin{proposition}\label{proposition: optimal set}
	 Consider \ref{mean-variance portfolio}, and assume we are given a candidate portfolio $\mathbf{x}$.  Then, $\mathbf{x} \in S(Q,\mathbf{c},A,\mathbf{b};r)$ if and only if there exists a $\mathbf{u}$ such that:
     \begin{align}
     \label{math: kkt}
	 \begin{array}{llll}
	      S(Q,\mathbf{c},A,\mathbf{b};r)=\{ & \mathbf{x}: A\mathbf{x} \geq \mathbf{b}, \mathbf{u} \in       \mathbb{R}_{+}^{m}, \\
	                                & \mathbf{u}^{T}(A\mathbf{x}-\mathbf{b}) = 0, \\
	                                & Q\mathbf{x} - r \mathbf{c} - A^{T}\mathbf{u} = 0.\}
	 \end{array}
	 \end{align}
	 \end{proposition}
	 \begin{proof}
	 If we seek to learn $r$, the optimal solution set for \ref{mean-variance portfolio} can be characterized by Karush-Kuhn-Tucker (KKT) conditions. For any fixed values of the data with $Q\in\mathbb{R}^{n \times n} \succeq 0$ and $\mathbf{c}$, the forward problem is convex and satisfies a Slater Condition. Thus, it is necessary and sufficient that any optimal solution $\mathbf{x}$ satisfy the KKT conditions. The KKT conditions are precisely the equations \eqref{math: kkt}. Here, $\mathbf{u}$ can be interpreted as the dual variable for the constraints in \ref{mean-variance portfolio}. 
	 \end{proof}

	 \begin{proposition}\label{proposition: linearization}
	 The equation $\mathbf{u}^{T}(A\mathbf{x}-\mathbf{b}) = 0$ in \eqref{math: kkt} is equivalent to that there exist $M > 0$ and $\mathbf{z} \in \{0, 1\}^{m}$ such that 
	 \begin{align}
	 \label{math: linearization}
	 \begin{array}{llll}
	      & \mathbf{u} \leq M\mathbf{z},  \\
         & A\mathbf{x} - \mathbf{b} \leq M(1-\mathbf{z}). \\ 
	 \end{array}
	 \end{align}
	 
	 \end{proposition}

	We next present a tractable reformulation of the update rule in \eqref{update rule} that constitutes the first main result of this paper. Applying Propositions \ref{proposition: optimal set} and \ref{proposition: linearization}, the update rule by solving \eqref{update rule} with the loss function \eqref{loss function for PO} becomes
	\begin{align}
	\label{kkt reformulation}
	\tag*{IPO}
	\begin{array}{llll}
            \min\limits_{r, \mathbf{x}, \mathbf{u}, \mathbf{z}} &\frac{1}{2}\|r - r_{t-1}\|^2 + \frac{\lambda}{\sqrt{t}}\| \mathbf{y}_{t} - \mathbf{x}\|^2\\
        	\; s.t. & A\mathbf{x}\geq \mathbf{b}, \\
        	\quad & \mathbf{u} \leq M\mathbf{z},  \\
        	\quad & A\mathbf{x} - \mathbf{b} \leq M(1-\mathbf{z}), \\
        	\quad & Q_{t}\mathbf{x} - r\mathbf{c}_{t} - A^T\mathbf{u} = 0,  \\
        	\quad & \mathbf{x} \in \mathbb{R}^{n} , \mathbf{u} \in \mathbb{R}_{+}^{m}, \mathbf{z} \in \{0, 1\}^{m}, 
	\end{array}
	\end{align}

 where $\mathbf{z}$ is a vector of binary variables used to linearize KKT conditions, and $M$ is an appropriate number used to bound the dual variable $\mathbf{u}$ and $A\mathbf{x} - \mathbf{b}$, and $\lambda$ is a learning step parameter. Clearly, \ref{kkt reformulation} is a mixed integer second order conic
program (MISOCP), and denoted by the \emph{Inverse Problem}.

\section{Validation and Experimental Setup}

Validation of learned risk preferences from the proposed data-driven approach is challenging because the inherent risks associated with existing portfolios are usually unknown. In this paper, we proposed two different frameworks to validate learned risk preferences. The first framework is called Forward-Inverse validation: we first generate dummy portfolio using known risk preference value, as the process of solving the \emph{Forward Problem}. Then, we learn the risk tolerance value as solving the \emph{Inverse Problem} and compare the difference (error) between the learned value and the known value. The main goal of this framework is to explore the best hyper-parameters to minimize the error. The second framework is setup for learning using real-world time-series portfolio data. We propose a novel procedure to align sequentially observed portfolios with real market price data, and decompose the data as sequential \emph{<price, portfolio>} pairs to learn risk preferences and their uncertainties via online learning.  
    
\subsection{Forward-Inverse Validation}\label{section: Forward-Inverse Validation}
\subsubsection{Point Estimation}
We benchmark hyperparameters $M$, $\lambda$ (also include $\varepsilon$ for factor-space algorithm presented in Section 6) of online inverse learning using grid search and Forward-Inverse validation. For each combination of hyper parameters, we sample a number of known risk tolerance values $r^s_{i}$ in a positive range (e.g.,[0.5,20]) in the forward problem to generate portfolio data. Then, we start with a number of different initial guess values $r^{g}_{j}$ from (e.g., [1,..,10]), using a specific group of hyper-parameters to estimate back the known risk tolerance value, denoted by $r^{e}_{ij}$, by solving the inverse problem. The total sum of square error between $r^{s}$ and $r^{e}$ for portfolio data generated by a specific $r^{s}$
\begin{align}
\label{squared error}
\frac{1}{|i|\times|j|} \sum_{i}\sum_{j} \frac{(r^{s}_{i} - r^{e}_{ij})^2}{(r_{i}^{s})^{2}}
\end{align}
is considered as average error of risk tolerance point estimation given a specific set of hyperparameters. The grid search space is defined as follows:
\begin{align}
M & \in [100, 500, 1000, 5000, 10000] \nonumber \\
\lambda & \in [100, 500, 1000, 5000, 10000] \nonumber \\
\varepsilon & \in [0.005, 0.01, 0.05, 0.1] \textrm{(only for factor space)} \nonumber
\end{align}
	
Our results show forward-Inverse validation using real market price data, the best combination of hyper-parameters can reduce the estimation errors smaller than $1 \times 10^{-9}$. 

\subsubsection{Ordered Risk Preference Learning}
We further generate a series of portfolios using a number of sorted risk tolerance values $r_{1} \leq,..., \leq r_{m}$. Our purpose is to validate whether the learned risk tolerance values can preserve the true order in Forward-Inverse benchmark using real market data.
\begin{align}
\{\mathbf{y}_{1}, ..., \mathbf{y}_{d} \} & \leftarrow \textrm{Forward Problem}(\{r_{1},...,r_{d}\}) \nonumber \\
\{\hat{r}_{1}, ..., \hat{r}_{d} \} & \leftarrow \textrm{Inverse Problem}(\{\mathbf{y}_{1},...,\mathbf{y}_{d}\}) 
\end{align}
	
The validation consists of three steps. First, we plot the efficient frontiers (EF) by increasing $r$ in the \emph{Forward Problems}. Each data point on the EF curve represents a (risk, profit) pair determined by a predefined $r$. Smaller $r$ are positioned closer to origin points on the EF curve, whereas larger values locate far from the origin on the EF curve. Then, we uniformly sample a number of (e.g., 30) risk tolerance values along the curve from small to large, and use the optimal allocation results as inputs to solve the \emph{Inverse problems}. Then, the estimated values are compared with true values. We compare both the orders and the absolute values along diagonal lines. Ideally, good estimations should align most points on the diagonal (Figure \ref{figure:Order Estimation of Risk Tolerance Values in Sector Space}).      

	\begin{figure}
    \centering 
	\begin{subfigure}
		\centering
		\includegraphics[width=.2\textwidth]{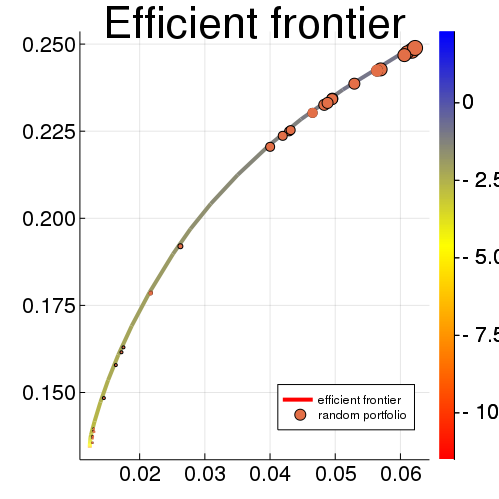}
	\end{subfigure}%
	\begin{subfigure}
		\centering
		\includegraphics[width=.2\textwidth]{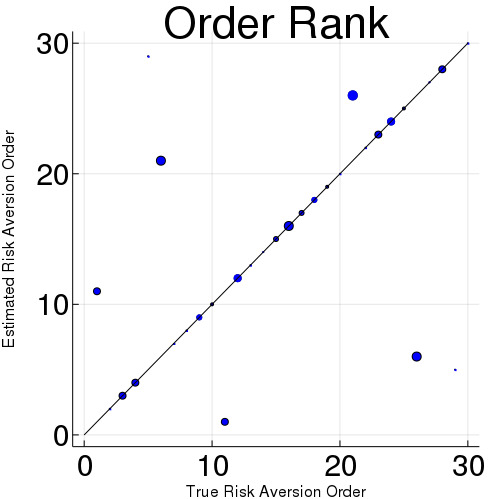}
	\end{subfigure}	
	\caption{Order Estimation of Risk Tolerance Values. The plot is generated using VFINX asset data observed from Jan 1st 2000 to March 31st 2015 aggregated to sector level. First, the efficient frontier curve in the left figure is generated by solving \emph{forward problem} using 100 evenly sampled $r$ values on logspace from -5 to 1, corresponding to $r=[10^{-5},...,10]$. The coordinates of curve are (Risk, Profit) values represented by ($\mathbf{x}^{T}Q\mathbf{x},\mathbf{c}^{T} \mathbf{x}$) where $\mathbf{x}$ are obtained by solving Forward problem. The color of curve is mapped to the logarithm of risk tolerance values as indicated in color map. In the same log space, 30 random risk tolerance values, indicated by red circles, are sampled in uniform distributions. The sizes of circles are proportional to their values. Then, portfolios generated by these samples in \emph{Forward Problem} are used as observations to estimate back $r$ in \emph{inverse problem}. The right figure compares the order of true values (x-axis) with the order of estimated values (y-axis).}
	\label{figure:Order Estimation of Risk Tolerance Values in Sector Space}
\end{figure}	
	

\subsection{Time-Series Portfolios Experimental Framework}\label{section:Time-Series Portfolios Experimental Framework}

\begin{figure}[ht]
	\includegraphics[width=.4\textwidth]{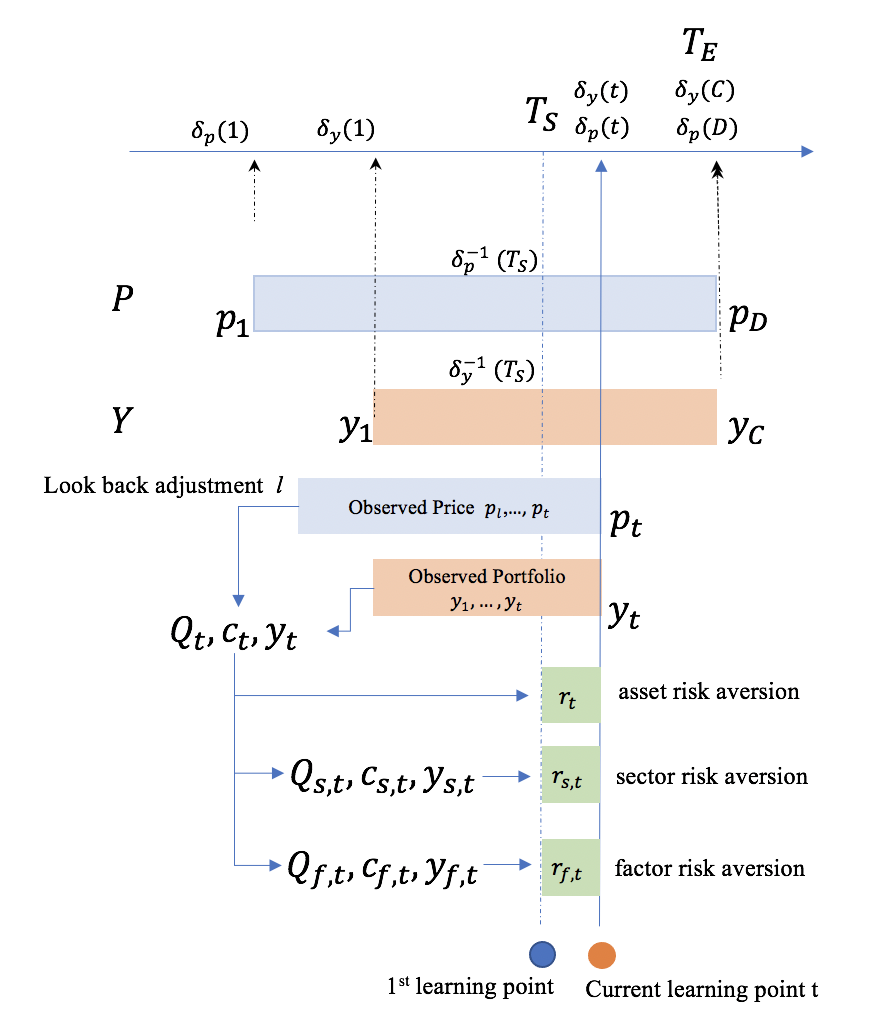}
	\caption{General framework for time-varying portfolio observations and risk preference learning.}%
	\label{figure: time-varying observations}
\end{figure}	

\subsubsection{Time-series Portfolio Data $Y$}
\noindent We define asset-level portfolio holdings $Y=\{\mathbf{y}_{1},...,\mathbf{y}_{C}\}$ collected from starting date $\delta_{y}(1)$ till ending date $\delta_{y}(C)$, where $\delta_y(\cdot)$ is a function that maps portfolio data indices to dates. $Y$ may be collected at different frequencies (e.g., daily, quarterly, etc.), and each $\mathbf{y}_{t} \in Y ( t=1,...,C)$ represents a portfolio vector consists of $n$ assets.  

\subsubsection{Asset Price Data $P$}
\noindent Analogously, asset price market data $P=\{\mathbf{p}_{1},...,\mathbf{p}_{D}\}$ is usually collected using some public APIs (e.g., AlphaVantage\cite{alphavantage}, Yahoo Finance) from starting date $\delta_{p}(1)$ to ending date $\delta_{p}(D)$. In our approach, $P$ is recorded at trading-day daily bases and the raw data includes market opening price and market closing price. On that basis, we use a 253-day rolling window (trading days) to calculate yearly profit ratios using the current day market closing price minus previous year's market opening price, then normalized by the later to obtain an annual return sequence $P'=\{\mathbf{p}'_{1},...,\mathbf{p}'_{D}\}$. Also, we usually collect more historical price data than portfolio data, therefore $D$ is usually larger than $C$, and $\delta_{p}(1)$ is earlier than $\delta_{y}(1)$. 

\subsubsection{Align Portfolio data with Asset Price data}
\noindent The collections of $Y$ and $P$ then will be aligned on the same time horizon. Lets define $t$ is an observation day, $\delta_{y}^{-1}(t)$ is a function to map $t$ to the index number in $Y$, and $\delta_{p}^{-1}(t)$ maps $t$ to index in $P$. Thus, after aligning, we can layout $P$ on the time horizon from $\delta_{p}(1)$ to $\delta_{p}(D)$, and $Y$ from $\delta_{y}(1)$ to $\delta_{y}(D)$. At a observation point $t$ (if $t$ is overlapped with portfolio and price data spans), we have observed portfolio sequence  $Y_{t}=\{\mathbf{y}_{1},...,\mathbf{y}_{\delta_{y}^{-1}(t)}\}$, and price sequence $P_{t}=\{\mathbf{p}_{1},...,\mathbf{p}_{\delta_{y}^{-1}(t)}\}$. We treat those as historical price data and portfolio data till observation time $t$, and use them to learn risk preference at time $t$.    

\subsubsection{Look-back window and Uncertainty Measure} Usually it is an arbitrary choice to determine the starting date of price $\delta_{p}(1)$ included in analysis. Sometimes this choice could have impacts on the forward portfolio construction results as well as inverse risk preference learning. To simulate this uncertainty, we introduce a look back window $l$ as adjustment of the observation of historical market signals $P$, given by $P_{t}=\{\mathbf{p}_{l},...,\mathbf{p}_{\delta_{y}^{-1}(t)}\}$. When $l=1$, we end up using all collected historical price data, and we can increase $l$ with a certain step length to $L$. For example, if $L=100$, it means we truncate out almost half year of collected price data, and use the remaining data as observations to learn risk preference. In our experiment, we usually try different $l$ values from 1 to $L$, and report the mean and standard deviations of learned risk preferences.

\subsubsection{Calculating time-varying covariance matrices and expected portfolio returns} At each observation time $t$, the covariance matrix and expected return of $n$ assets $\mathbf{c}_{t}$ is calculated as
\begin{align}
    Q_{t} &= \textrm{cov}([\mathbf{p}'_{l},...,\mathbf{p}_{t}],[\mathbf{p}'_{l},...,\mathbf{p}_{t}]),  \\
    \mathbf{c}_{t} &= [c_{1,t},...,c_{i,t},...,c_{n,t}]^{T},
\end{align}
\noindent where $c_{i,t}=\textrm{mean}([p'_{i,l},...,p'_{i,t}])$ is the mean profit value of $i-$th asset in historical price data. 

\subsubsection{Learning time-varying risk preferences} Now we can define a date range called learning period, from $T_{s}$ to $T_{e}$ from the time horizon where $P$ and $Y$ overlap. Usually we select a $T_{s}$ later than the first portfolio date $\delta_{p}(1)$ because we want to include multiple samples $\{y_{1},...,y_{\delta_{y}^{-1}(T_{s})}\}$ for online learning. Otherwise, if we set $T_{s}==\delta_{p}(1)$, then the first observation will have only one portfolio sample, and the learned risk preferences at the beginning could be very unstable because there is not enough samples to learn. After learning at $T_{s}$, we move to $T_{s}+1$ and include a new sample in our portfolio in observed sequence, and learn the next risk preference, and so on. If from $T_{s}$ to $T_{e}$ there are $T$ number of learning period, we obtain a time-varying risk preference curve. And with adjustments of look back window $l$, we repeat the process can simulate a confidence band of risk preferences.          

\subsubsection{Aggregation of $Q_{t}, Y_{t}, \mathbf{c}_{t}$:} When the sizes of portfolios are large ($n$>1000) such as mutual funds mentioned in Section 6, to increase the efficiency and stability of IPO ($n=500$ takes 80 CPU seconds, $n>1000$ takes more than 1000 seconds), also to enable easy comparison of portfolios with different sizes, we may group assets by sectors, or apply PCA on $Q_{t}$ to factorize portfolios, thus obtaining $Q_{s,t}, c_{s,t}, y_{s,t}$ in sector space and $Q_{f,t}, c_{f,t}, y_{f,t}$ in factor space respectively to learn risk preferences in Sector or Factor space. \\

\noindent So far we have introduced a general experimental framework setup for time varying portfolio risk preference learning. In next two sections, we will use this framework to explain our data sets and experimental settings.

\section{Case Study 1: Robotic portfolios}

	\begin{figure*}[!t]
	    \centering
		\includegraphics[width=.95\textwidth]{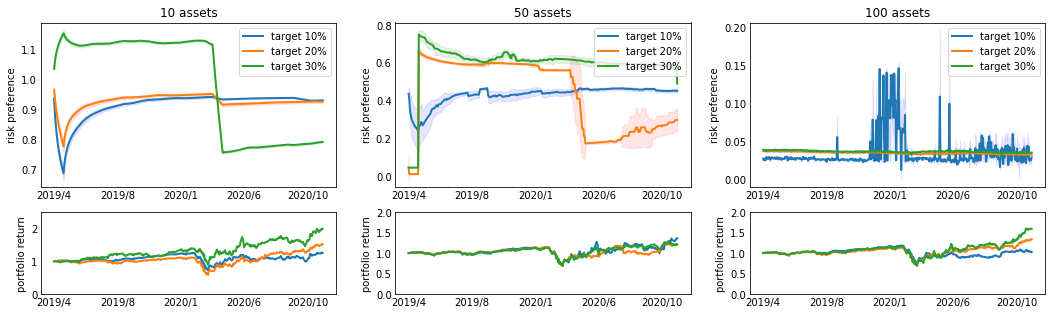}
		\caption{Comparison of robotic portfolio risk preference Values. Figures at the top are learned risk preference values, and figures at the bottom are profit returns.}%
		\label{figure: Comparison of robotic portfolio risk preference Values}
	\end{figure*}
	
\noindent We apply proposed method on a series of actively-managed portfolios generated from an in-house portfolio construction platform. This platform utilizes a Deep Reinforcement Learning (DRL) model \cite{robo-advising} to automatically construct portfolios for specified target investment goals. The DRL model addresses a similar objective as the forward problem \ref{mean-variance portfolio} used in our risk preference learning approach, but formulates it in multi-period setting \cite{multiperiod} and generates auto-balanced portfolios through deep reinforcement learning.	An important input parameter to specify in such a DRL model is the \emph{expected portfolio return}, which could be easily proved to have a one-to-one correspondence with the \emph{risk preference} in single-period mean-variance formulation (See Appendix \ref{equivalence}). Since specifying target return is more intuitive to investors than determining risk preference, the DRL system relies on extensions of the following equivalent formulation:
	
	\begin{align}
	\label{deep-reinforcement-problem}
	\tag*{PO-DRL} 
    \begin{array}{llll}
         \min\limits_{\mathbf{x} \in \mathbb{R}^{n}} &  \frac{1}{2} \mathbf{x}^{T} Q \mathbf{x}  \\
    	\;s.t. &  \mathbf{c}^{T} \mathbf{x} =  z, \\
    	       &   A \mathbf{x} \geq  \mathbf{b}     
    \end{array}
	\end{align}	
	
	\noindent where $z$ is the \emph{expected portfolio return}. In the multi-period optimization (MPO) setting \cite{multiperiod}, the dynamics of the portfolio rebalanced at time point $i=0,1,2,...,N-1$ is 
	\begin{align}
	\label{multi-period IPO}
	\tag*{MPO}
	\begin{array}{llll}
        W_{i+1} = \sum_{j=1}^n v_i^j\frac{S_{i+1}^j}{S_i^j} + W_i-\sum_{j=1}^n v_i^j, \quad i=0,\dots,N-1.
    \end{array}
	\end{align}
	\noindent where $W_i$ represents the wealth at the $i$-th fine-tuning time, $S_{i}^j$ represents the price of stock $j$ at $i$. The vector $\mathbf{v}_i=(v_i^1, v_i^2, \dots, v_i^n)^T$ gives the dollar value allocations among all the $n$ stocks in the portfolio at each fine-tuning time $i$.  The precise multi-period extension is hence 
    \begin{eqnarray}
    & & \min_{\mathbf{v}_i, i=0,\dots, N-1} \text{Var}[W_N], \nonumber \\
    & & \; s.t. \ \mathbb{E}[W_N]=1+z.\label{Multi_MV}
    \end{eqnarray}
    It is assumed above that the initial capital $W_0=1$. Our in-house DRL platform solves the problem in equation \eqref{Multi_MV} using deep neural networks and actor-critic RL methodologies. More information about the implemented DRL algorithms is available in \cite{robo-advising}.

    \textbf{Experimental Setup:} 
	 The DRL model is trained and validated using daily market price data of S\&P 500 stocks from 2017-04-01 to 2019-04-01, and generates daily-rebalanced portfolios from 2019-04-01 till 2020-12-10. DRL allows investor to specify the size of portfolio ($n$ assets), and also provides options allowing switching of assets or keeping the same assets. In our experiments, we turn off asset switching functions, and create three sizes of portfolios (10, 50, 100 assets) using three specified annual \emph{expected portfolio returns} (10\%, 20\%, 30\%) as investment target. The DRL model automatically rebalances the portfolios per each trading day, so for the whole test period it generated 9 different time-series portfolios (combing asset size with target profit) where each portfolio is composed by 425 daily portfolio vectors. Our goal is to learn the risk preferences of these robotic portfolios and compare our learning results with the investment target values that drive the DRL model. 
	
    Per discussions about experimental framework in Section \ref{section:Time-Series Portfolios Experimental Framework}, the portfolio starting date $\delta_{y}(1)$ is 2019-04-01, and ending date $\delta_{y}(C)$ is 2020-12-10. The price data starting date $\delta_{p}(1)$ is 2000-01-02, and ending date $\delta_{p}(D)$ is 2020-12-10. Because the portfolios are daily data, the price dates and portfolio dates are aligned on trading days. We set the observation period $T_{s}=\delta_{y}(1)$ and $T_{e}=\delta_{y}(C)$. The look back period adjustment $l$ are ten values from $\{1,11,21,...,101\}$. Hyper-parameters are optimized using as single point forward-backward validation (Section \ref{section: Forward-Inverse Validation}). 
	
	\textbf{Results: } Our results show that even if these robotic portfolios are constructed using a different (more complicated, across-time) objective, the proposed inverse optimization process can still learn the ordinal relationships of important decisional inputs governing the investment procedures. As shown in Figure \ref{figure: Comparison of robotic portfolio risk preference Values}, the learned daily risk preference curves corresponding to higher target returns are generally larger than those associated with lower target returns in all portfolios. Most of the time, the learned risk preferences are stable across the investment period, but significant volatility occurs on some portfolios around March 2020, which is understandable because it reflects the market movement that happened during that period. If we first look at the return of those robotic portfolios (bottom figures), all portfolios have positive returns before turning negative till March 2020, and afterward, they all recover to be profitable. Small portfolios composed of 10 assets achieve higher returns at Dec 2020 (about 100\% actual profits for those targeting 30\% annual returns), whereas larger portfolios have slightly lower returns (about 50\% actual profits for 50-asset and 100-asset portfolios targeting 30\% yearly returns). Linking the investment targets of portfolio returns back to learned risk preferences at the top figures, they are generally consistent at the ordinal level, which is also understandable because of the one-to-one correspondence between target return and risk preference in the DRL's multi-period objective function. It is also noticeable that around March 2020 three risk preference curves vibrate drastically (10-asset target 30\%, 50-asset target 20\% and 100-asset target 10\%), which might be due to the \emph{regime switch} in DRL model because of the underlying changes of market data, as well as the rigidness of the long-only allocation constraint.
	
	\textbf{Discussions: }Besides the promising results, it is also noticeable that the sizes of portfolios impact risk preferences learned by the proposed method. According to our experiments, larger portfolios usually lead to smaller risk preferences. Apart from the coincidence in portfolio theory that diversification is preferable in investment (which may imply lower risk), we believe this is more due to the numerical issue in the \ref{mean-variance portfolio} objective: when the dimension of $x$ turns higher under the constraints that its L1-norm should be 1, the first quadratic term $x^{T}Qx$ becomes smaller, and thus impacts the risk preference parameter $r$ in the second term, which is basically a regularization parameter, to turn smaller. Such a property causes an issue when we compare risk preferences across portfolios of different sizes, and we will address this issue in the next case study.

\section{Case Study 2: Mutual Funds}
In this case study, we will explore risk preferences learning on seven time-series mutual fund portfolio holdings. To compare portfolios of different sizes, we propose two extensions of the original asset-level algorithm by aggregating high dimensional portfolios into sector space and factor space. The sector space algorithm utilizes attribute knowledge of each asset and groups them into 11 industrial sectors. Alternatively, we can apply Principal Component Analysis (or other factor analysis) on the covariance matrix and learn risk preferences on factorized portfolios. Due Readers interested in these extensions can read the detailed derivations at this paper's full version at arXiv. 

\textbf{Mutual Funds:} VFINX and SWPPX are S\&P 500 funds broadly diversified in large-cap market. FDCAX's portfolio mainly consists of large-cap U.S. stocks and a small number of non-U.S. stocks. VHCAX is an aggressive growth fund invests in companies of various sizes that the managers believe will grow in time but may be volatile in the short-term. VSEQX is an actively managed fund holding mid- and small-cap stocks that the advisor believes have above-average growth potential. VSTCX is also an actively managed fund that offers broad exposure to domestic small-cap stocks believed to have above-average return potential. LMPRX mainly invests in common stocks of companies that the manager believes will have current growth or potential growth exceeding the average growth rate of companies included in S\&P 500.

The reasons of analyzing mutual fund portfolios are threefold. First, mutual fund holdings are freely accessible public data and it is easier to discuss and interpret results using public data than clients' private data. Second, mutual funds are usually constructed by tracking industry capitalization indices, or managed by fund managers. Thus, they can be considered as optimal portfolios constructed through rational decision making. Moreover, risk measurements of mutual funds are common knowledge reported by several well-known metrics and estimations with our methods can be directly validated by these metrics. Third, mutual funds usually are diversified on large numbers of assets and they fit the underlying MPT used in our model. Moreover, high-dimensional portfolios are challenging to optimize in the inverse optimization problem and they really test the efficiency of our approach. 

\textbf{Experimental Setup:} We collect 10 years of asset-level portfolio holdings of the aforementioned seven mutual funds from quarterly reports publicly available on SEC website, starting from March 2010 ($\delta_{y}(1)$) to Dec 2020($\delta_{y}(C)$). We also collect twenty years of historical asset level price data from July 1st, 2001 ($\delta_{p}(1)$) to Dec 21st, 2020 ($\delta_{p}(D)$). The learning period is set from March 2015 ($T_{s}$, corresponding to $\delta_{y}^{-1}(T_{s})=21$) till Dec 2020 ($T_{e}$, corresponding to $\delta_{y}^{-1}(T_{e})=44$). Market price data is aggregated to monthly level to align with quarterly portfolio. The adjustment parameter of look back $l$ is evenly selected from $1,3,5,...,11$ (months).

To further validate the estimated risk tolerance values, we compare them with two standard metrics evaluating investment risk used in financial domain: \textbf{Inverse of Sharpe Ratio} and \textbf{Mutual Fund beta values }. 
	
	\begin{figure*}[!t]
	    \centering
		\includegraphics[width=.9\textwidth]{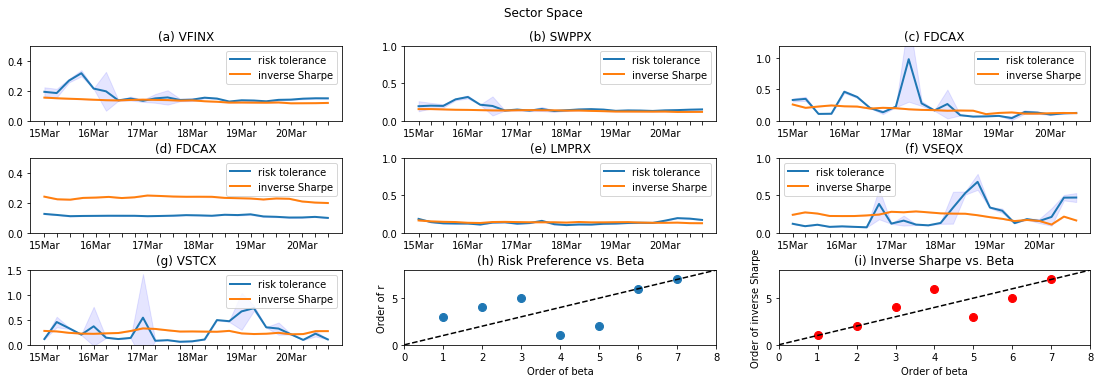}
        
		\includegraphics[width=.9\textwidth]{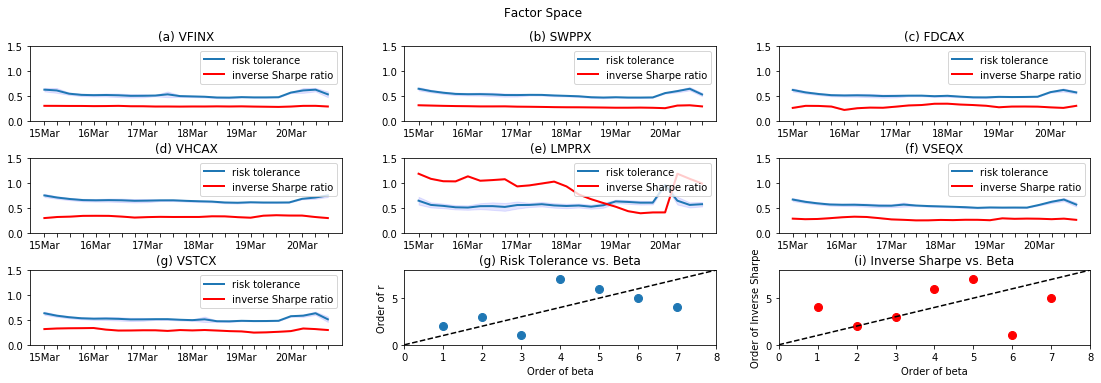}
		\caption{Comparison of Mutual Fund Risk Tolerance Values in Sector space and Factor space. Estimated risk tolerance parameters (blue curves) of seven mutual funds are illustrated from (a) to (g). For each fund, 24 time-varying risk tolerance values are estimated, corresponding to 24 quarterly observed portfolio holdings from March 2015 to December 2020. In Factor Space, we perform eigenvalue decomposition on $Q_{t}$ and select 5 eigenvectors corresponding to 5 largest eigenvalues as factors. In figure (h) of both spaces, we rank the seven mutual funds by the order of average risk tolerance values estimated over all 24 quarters, from small to large, and compare them with the order of beta values shown in Table \ref{tab:Beta values of mutual funds studied in our approach}. In figure (g), we compare the order ranked by inverse Sharpe ratios with order ranked by beta values.}%
		\label{figure:Comparison of Mutual Fund Risk Tolerance Values in Sector Space}
	\end{figure*}
	
	
	 \textbf{Comparison with Inverse of Sharpe Ratio:} In finance, the Sharpe ratio \citep{Sharpe1965} measures the performance of a portfolio compared to a risk-free asset, after adjusting the risk it bears. It is defined as the difference between the expected return of the asset return $E(R_{a})$ and the risk-free return $E(R_{f})$, divided by the standard deviation of the investment:
	\begin{align}
	\label{math:Sharpe ratio}
	    S_{a}= \frac{E(R_{a}) - E(R_{f})}{\sigma}
	\end{align}

	 Note that the Sharpe ratio is proportional to \emph{risk-aversion}. In our paper, we use a modified \emph{Inverse Sharpe ratio} to compare with learned \emph{risk tolerance} values:
	\begin{align}
	\label{math:Inverse Sharpe ratio}
	    IS_{a,t}= \frac{\mathbf{x}_{a,t}^{T}Q_{a,t}\mathbf{x}_{a,t}}{\mathbf{c}_{a,t}^{T} \mathbf{x}_{a,t}}
	\end{align}
     In \eqref{math:Inverse Sharpe ratio}, $\mathbf{x}_{a,t}, Q_{a,t}, \mathbf{c}_{a,t}$ are variables defined over asset-level, and change over time. To avoid redundant discussions, we analogously derive similar Inverse Sharpe ratios in sector and factor space.
     
    
     According to Figure \ref{figure:Comparison of Mutual Fund Risk Tolerance Values in Sector Space}, estimated $r$ are more volatile over time in sector space, and the confidence intervals are generally wider than inverse Sharpe ratios. However, in factor space, estimated $r$ becomes smooth and Inverse Sharpe ratio is mostly lower than proposed risk preference, where the main reason is factors represent dominant volatility. 
     

	 \textbf{Comparison with Mutual Fund beta values:} In Capital Asset Pricing Model (CAPM) \citep{Treynor1961MarketVT,Sharpe64,Lintner65,Mossin1966}, the \emph{beta coefficient} of a portfolio is a measure of the risk arising from exposure to general market movements as opposed to idiosyncratic factors, given by 
	\begin{align}
	    \beta_{a}= \frac{E(R_{a}) - E(R_{f})}{E(R_{m}) - E(R_{f})} ,
	\end{align}
	 where $E(R_{m})$ is the expected return of market (e.g., benchmark portfolio) and others are the same as defined in \eqref{math:Sharpe ratio}. The \emph{beta} measures how much risk the investment is compared to the market. If a stock is riskier than the market, it will have a \emph{beta} greater than one. If a stock has a \emph{beta} of less than one, the formula assumes it will reduce the risk of a portfolio.
	
	 The \emph{beta} values of mutual funds are public data and change over time. We list the recent two beta values in Table \ref{tab:Beta values of mutual funds studied in our approach}.

	\begin{table}
		\tiny
		\centering
		\begin{tabular}{c|c|c|c|c|c|c|c}
			\hline
			                     & VFINX & SWPPX & FDCAX & VHCAX & VSEQX & VSTCX & LMPRX \\
			\hline                     
			$\beta$(2019) & 1.0   & 1.0   & 1.03  & 1.14  & 1.17  & 1.22  & 1.17 \\
            $\beta$(2020) & 1.0   & 1.0   & 1.01  & 1.13  & 1.30  & 1.38  & 1.07\\
			\hline
		\end{tabular}
		\caption{Beta values of mutual funds studied in our paper}
		\label{tab:Beta values of mutual funds studied in our approach}
	\end{table}	
    
    We sort the averaged risk-tolerance values of 7 mutual funds over time from small to large, and compare their orders with the order of beta values averaged from Table \ref{tab:Beta values of mutual funds studied in our approach}. In Figure \ref{figure:Comparison of Mutual Fund Risk Tolerance Values in Sector Space} (h), most of the scatter plots are close to diagonal, indicating that the order of estimated risk tolerance is consistent with the order of beta values. Therefore, risk indicators estimated by proposed methods are intuitively rational because they are consistent with common knowledge in the financial domain. 
	 
	\textbf{Readiness for automated and personalized portfolio construction} 
    Our risk tolerance parameters can be directly used as inputs in the mean-variance framework to drive automated portfolio advice, known as Robo-advising. Most existing Robo-advising systems construct risk preference parameters according to \emph{one-time interaction} with the client \citep{Capponi2019PersonalizedRE}, typically profiled based on client's financial objectives, investment horizons, and demographic characteristics. Other approaches assume the client knows her risk tolerance parameter at all times, but only communicates it to the Robo-advisor at specific updating times\citep{Capponi2019PersonalizedRE}. With our approach, we could anchor investor client's risk at an appropriate level compared to some benchmark portfolios. Our experiments show that the time-varying risk preference of two S\&P 500 mutual funds, which are considered a market benchmark portfolio, range from 0.2 to 0.5 in Sector Space. We also illustrate that some actively managed funds (VSEQX, VSTCX) have larger risk preference values than 0.5. For example, if the client is aware that S\&P 500 fund has risk tolerance $r=0.2$ and an active fund has risk tolerance $r=0.5$, she may prefer to have a personalized risk tolerance that is between those of the S\&P 500 fund and active fund at current market period. The Robo-advisor can directly suggest a personalized risk tolerance score $r=0.4$ and start with a personalized portfolio construction process. Later, when the market has changed, risk tolerance values of S\&P 500 fund and active funds may also change accordingly. The Robo-advisor should update the process with a new parameter to reflect the time-varying risk preference.

\section{Conclusions}
In this paper, we present a novel approach to learn risk preference parameters from portfolio allocations. We formulate our learning task in a generalized inverse optimization framework. We consider portfolio allocations as outcomes of an optimal decision-making process governed by mean-variance portfolio theory, thus risk preference parameter is a learnable decision variable through IPO. The proposed approach has been validated and exhibited in two real problems that consists of a 18-month of daily robotic investment portfolios and ten-year history of quarterly portfolio holdings collected in seven mutual funds. The proposed method can be applied to individual investors' portfolios to understand their risk profiles in personalized investment advice. It could be deployed as real-time reference indices in Robo-advising systems that facilitate automation and personalization.

	
\bibliographystyle{ACM-Reference-Format}
\bibliography{sample-base}

\clearpage
\newpage

\section{Supplementary Material A: Learning Risk in Sector and Factor Space}\label{section:Investment Risk Tolerance Learning}
 In this section, we introduce reformulations of \ref{kkt reformulation} for sector-based and factor-based risk preference learning in details.  
\subsection{Sector Space}
\subsubsection{Forward Problem}
Similar to \ref{mean-variance portfolio}, we consider the following mean-variance portfolio optimization problem in sector space:
\begin{align}
\label{forward problem of po}
\tag*{PO-Sector}
\begin{array}{llll}
     \min\limits_{\mathbf{x}} &  \frac{1}{2} \mathbf{x}_{s}^{T} Q_{s} \mathbf{x}_{s} - r_{s} \mathbf{c}_{s}^{T} \mathbf{x}_{s}  \\
     \; s.t.  &  A_{s} \mathbf{x}_{s} \geq  \mathbf{b}_{s},
\end{array}
\end{align}
 where $Q_{s}$ is the sector level variance-covariance matrix, $\mathbf{x}_{s}$ is the portfolio allocation aggregated in sector space, $\mathbf{c}_{s}$ is the expected sector level portfolio return, $r_{s}$ is the risk tolerance factor (in sector space), $A_{s}$ is the structured constraint matrix, and $\mathbf{b}_{s}$ is the corresponding constraint value. 
	 
	 Similar to Section 4.2.5, time-varying $Q_{s,t}$ and $\mathbf{c}_{s,t}$ are defined as follows:
	\begin{align}
	\begin{array}{llll}
 	Q_{s,t} & = \textrm{cov}([\mathbf{p}_{s,l},...,\mathbf{p}_{s,t}],[\mathbf{p}_{s,l},...,\mathbf{p}_{s,t}]), \\
 	\mathbf{p}_{s,t} & = [\hat{p}_{i,t}]_{i \in \mathcal{S}_{i}} \\
	\hat{p}_{i,t} & = \sum\limits_{j \in \mathcal{S}_{i}} w_{j,t} p_{j,t}, \\
	c_{i,t} & = \textrm{mean} ([\hat{p}_{i,t-w},...,\hat{p}_{i,t}]), \\
	\mathbf{c}_{s,t} & = \{c_{1,t},...,c_{i,t},...,c_{|\mathcal{S}|,t}\},
	\end{array}
	\end{align}
	 where $\mathbf{p_{s,t}}$ is vector of sector-level returns, $\mathcal{S}_{i}$ is the set of assets in $i$-sector, $w_{j,t}$ is asset holding weights belonging to $\mathcal{S}_{i}$, and $\hat{p}_{i,t}$ is time-varying sector monthly profit calculated by weighted sum of asset level profit in that sector. 

	\subsubsection{Inverse Problem}\label{section: ipo in sector space}
	The online inverse problem of learning sector-level risk tolerance $r_{s}$ given $\mathbf{y}_{s,t}, Q_{s,t}, \mathbf{c}_{s,t}$ is formulated as 
	\begin{align}
	\label{full formulation of ipo}
	\tag*{IPO-Sector}
	\begin{array}{llll}
	    \min\limits_{{r}_{s}, \mathbf{x}_{s} \in \Theta} &\frac{1}{2}\|{r}_{s}-r_{s,t}\|^2 + \eta_t\| \mathbf{y}_{s,t} - \mathbf{x}_{s}\|^2\\
    	\; s.t. \quad & A_{s} \mathbf{x}_{s}\geq \mathbf{b}_{s}, \\
    	\quad & \mathbf{u} \leq M\mathbf{z}, \\
    	\quad & A_{s} \mathbf{x}_{s} - \mathbf{b}_{s} \leq M(1-\mathbf{z}), \\
    	\quad & Q_{s,t}\mathbf{x}_{s} - r_{s} \mathbf{c}_{s,t} - A_{s}^T \mathbf{u} = 0, \\
    	\quad & \mathbf{x}_{s} \in \mathbb{R}^{m} , \mathbf{u} \in \mathbb{R}_{+}^{m}, \mathbf{z} \in \{0, 1\}^{m}. 
	\end{array}
	\end{align}
	The observed portfolio $\mathbf{y}_{s,t}$ is now aggregated to sector level.
	
	\subsubsection{Sector Risk Tolerance Learning Algorithm}
	
	 Algorithm \ref{algorithm: OL} illustrates the process of applying online IOP to learn risk tolerance using sector level data. The number of iteration $T$ varies by observed portfolio. In the mutual fund case study, at $\mathbf{y}_{20}$, we have 20 historical quarterly portfolios so $T=20$. $T$ increases each quarter to $T=44$ at December 2020.
	
	\begin{algorithm}
		\caption{Risk Tolerance Learning from Sector Portfolio}
		\label{algorithm: OL}
		\textbf{Input:} (time-series portfolio and price data) $ {Y}_{t}, P_{t}$\\
		\textbf{Initialization:} $r_{0}$ (guess), $\lambda$, $M$ (hyper-parameter)   
		\begin{algorithmic}[1]
			\For{$t=1$ to $T$}  
			\State receive ($Y_{t}, P_{t}$)   
			\State $(\mathbf{y}_{s,t}, \mathbf{p}_{s,t}) \gets (\mathbf{y}_{t}, \mathbf{p}_{t})$
			\State $(Q_{s,t}, \mathbf{c}_{s,t}) \gets \mathbf{p}_{s,t}$    
			\State $\eta_{t} \gets \lambda t^{-1/2}$  \Comment{get updated $\eta_{t}$ by decaying factor}
			\State Solve $r_{s,t}$ as \ref{full formulation of ipo} in Section \ref{section: ipo in sector space}
			\State $r_{s,t+1} \gets r_{s,t}$ \Comment{update new guess of $r$}
			\EndFor
		\end{algorithmic}
		\textbf{Output:} Estimated $r_{s}$ 
	\end{algorithm}

	\subsubsection{Hyper-parameter and constraints}
	 Solving the inverse problem in \ref{full formulation of ipo} requires hyper-parameters $M$ and $\eta_{t}$. The constraints $A_{s} \mathbf{x}_{s}\geq \mathbf{b}_{s}$ correspond to
	\begin{align}
	& 0 \leq x_{i} \leq 1, \ \  i \in [11], \\
	& \sum_{i=1}^{11} x_{i} = 1,
	\end{align}
 	 as we do not consider short-selling positions. 
Specifically, the structural constraint coefficient matrix $A_{s}$ is a $22 \times 11$ matrix, and $\mathbf{b}_{s}$ is a $22 \times 1$ vector:
	\begin{align}
	\label{math: A and b}
	A_{s} = \begin{bmatrix}
	1 & & & &     \\
	& 1 & & &     \\
	& & \ddots & &\\
	& & & 1 &     \\
	& & & & 1   \\
	-1 & & & &     \\
	& -1 & & &     \\
	& & \ddots & & \\
	& & &  -1 &     \\
	& & & &  -1   \\
	1 & 1 & \hdots & & 1   \\
	-1 & -1 & \hdots & & -1   \\
	\end{bmatrix},  \ \ & 
	\mathbf{b}_{s} = \begin{bmatrix}
	0 \\
	0 \\
	\vdots \\
	0 \\
	0 \\
	-1 \\
	-1 \\
	\vdots \\
	-1 \\
	-1 \\
	1 \\
	-1 
	\end{bmatrix}.
	\end{align}
	
	\subsection{Factor Space}
	\subsubsection{Forward Problem}
	 Without abuse of notation, we refer to the asset-level covariance matrix as $Q$, and the asset-level return as $\mathbf{c}$. We perform eigendecomposition of the asset-level covariance matrix $Q$ such that 
	\begin{align}
	\label{math:eigendecomposition}
	 Q \approx F \Sigma F^T,
	\end{align}
	 where $\Sigma$ is a $K \times K$ diagonal matrix of largest $K$ eigenvalues, $F$ is a $N \times K$ matrix of eigenvectors. 
	 The relationships between asset-level allocation $\mathbf{x}$ and factor-level allocation $\mathbf{x}_{f}$ are: 
	\begin{align}
	\label{math:factor-level allocation 1}
	F^{T}\mathbf{x} = \mathbf{x}_f,  \\
	\label{math:factor-level allocation}
	F^{T}\mathbf{c} = \mathbf{c}_f,
	\end{align}
	where \eqref{math:factor-level allocation 1} is projecting asset allocations to factor space, and \eqref{math:factor-level allocation} follows because $\mathbf{x}_f$ is a linear combinations of allocations in $\mathbf{x}$, thus the expectation of returns follows the same linear combination $E[\mu_{i} x_{i}]=\mu_{i}E[x_{i}]$. 
	 
	 \begin{proposition}\label{proposition: factor space po}
	 The forward problem in factor space is 
    \begin{align}
	\label{forward problem of po in factor space}
	\tag*{PO-Factor}
	\begin{array}{llll}
	     \min\limits_{\mathbf{x}_{f}} & \  \frac{1}{2} \mathbf{x}_f^T\Sigma \mathbf{x}_f - r_{f} \mathbf{c}_{f}^T \mathbf{x}_f\\
         \; s.t. \quad & AF \mathbf{x}_f\geq \mathbf{b}, 
	\end{array}
	\end{align}
	 where $\mathbf{c}_{f}=F^{T}\mathbf{c}$. 
	 \end{proposition}
	 \begin{proof}
    	 Applying \eqref{math:eigendecomposition} and \eqref{math:factor-level allocation} to the objective function in the first term of \ref{mean-variance portfolio} yields
    	\begin{align}
    	\label{math: replace objective function by decompostion}
    	\begin{array}{llll}
    	     &\frac{1}{2} \mathbf{x}^{T} Q  \mathbf{x} \\    	
        	= & \frac{1}{2} \mathbf{x}^T F \Sigma F^{T} \mathbf{x}  \\
        	= & \frac{1}{2} \mathbf{x}_f^T \Sigma \mathbf{x}_f ,
    	\end{array}
    	\end{align}
    	Then, for the new projected portfolio $\mathbf{x}_{f}$, the second term of \ref{mean-variance portfolio} can be rewritten as the new term $r_{f} \mathbf{c}_{f}^{T} \mathbf{x}_{f}$.
    	Notice that,
    	\begin{align}
    	\label{math: error of decompostion}
    	\begin{array}{llll}
    	    & r_{f} \mathbf{c}_{f}^{T} \mathbf{x}_{f} \\
    	    = & r_{f} (F^{T} \mathbf{c})^{T} F^{T} \mathbf{x} \\
    	    = & r_{f} c^{T} FF^{T} \mathbf{x} ,
    	\end{array}
    	\end{align}
    	where $FF^{T}$ is an Identity matrix only when $F$ contains all the eigenvectors ($K=N$), and then the new objective in \ref{forward problem of po in factor space} is equivalent to \ref{mean-variance portfolio} and $r_{f}$ is the same as $r$. However, if $K<N$, $FF^{T}$ is not identity matrix, so $r_{f}$ is also an approximation of $r$ in factor space.   
	 \end{proof}

	 In our experiments, $Q$ is a 2236 $\times$ 2236 covariance matrix of all underlying assets included in seven mutual funds. All mutual fund holdings are represented by dimensional vector $\mathbf{y}$, where non-holding assets are 0. The number of factors K in all our experiments is set to 5. 
	
	\subsubsection{Inverse Problem}\label{section:iop for factor space}
	 The structure of the inverse optimization problem in factor space is similar to that in asset space and sector space, and the main difference occurs on the constraint matrix $AF$. Since $Q$ changes over time, the decomposed $\Sigma$ and $F$ also are variables depending on $t$. 
	 \begin{theorem}
	 The inverse optimization problem in factor space is
	 \begin{align}
	\label{full formulation of ipo: factor space}
	\tag*{IPO-Factor}
	\begin{array}{llll}
	     \min\limits_{r_{f}, \mathbf{x}_{f} \in \Theta} &\frac{1}{2}\|{r}_{f} - r_{f,t} \|^2 + \eta_t\| \mathbf{y}_{f,t} - \mathbf{x}_{f}\|^2\\
    	\; s.t. \quad & AF_{t} \mathbf{x}\geq \mathbf{b}_{f},  \\
    	\quad & \mathbf{u} \leq M\mathbf{z},  \\
    	\quad & AF_{t} \mathbf{x}_{f} - \mathbf{b}_{f} \leq M(1-\mathbf{z}), \\
    	\quad & \Sigma_t \mathbf{x}_{f} - r_f \mathbf{c}_{f,t} - (AF_{t})^T\mathbf{u} = 0,  \\
    	\quad & \mathbf{c}_{f,t} \in \mathbb{R}^{K}, \mathbf{x}_{f} \in \mathbb{R}^{K} , \mathbf{u} \in \mathbb{R}_{+}^{K}, \mathbf{z} \in \{0, 1\}^{K}.    
	\end{array}
	\end{align}
	 \end{theorem}
	 \begin{proof}
	 The theorem is a direct result of applying Proposition \ref{proposition: factor space po} and KKT conditions.
	 \end{proof}

    \begin{remark}
     The new constraint matrix $AF_{t}$ in \ref{full formulation of ipo: factor space} is no longer a sparse structural matrix as it had in \ref{math: A and b}. $AF_{t}$ now becomes a dense matrix:
    	\begin{align}
    	AF_{t} = \begin{bmatrix}
    	F_{t}^{T} \\
    	-F_{t}^{T} \\
    	\sum_{j=1}^{n} f_{j,t}  \\
    	-\sum_{j=1}^{n} f_{j,t}   \\
    	\end{bmatrix},  \ \ & 
    	\mathbf{b}_{f} = \begin{bmatrix}
    	\varepsilon{\mathbf{1}} \\
    	-{\mathbf{1}} \\
    	1 \\
    	-1 
    	\end{bmatrix}.
    	\end{align}
    	 where $F_{t}$ is the $N \times K$ eigenvector matrix obtained in \eqref{math:eigendecomposition}, and $f_{j,t}$ are its $j$-th column. Notice that the first vector in $\mathbf{b}_{f}$ has changed from vector of all zeros to a vector of constant value $\varepsilon$. The reason is that solving \ref{full formulation of ipo: factor space} with a dense constraint matrix $AF_{t}$ is very challenging, and the linear equations in $AF_{t}$ can pose conflicting constraints and cause the problem infeasible. In practice, we find that relaxing one side of constraints from $0 \leq AF_{t}\mathbf{x} \leq 1$ to $\varepsilon \leq AF_{t}\mathbf{x} \leq 1$, where $\varepsilon$ is a small positive constant, can make the problem much easier to solve. Under this factor space setting, the allocation weights $\mathbf{x}_{f}$ on factors can be positive or negative values whose absolute values are larger than 1. When re-projecting them to the original asset space $F_{t}\mathbf{x}_{f}$ as \eqref{math:factor-level allocation}, however, all weights are still between $\varepsilon$ and 1, and their sum is equal to 1.	
    \end{remark}

	\subsubsection{Hyper-parameter and constraints}
	
	 Solving factor space inverse problem requires three hyper-parameters: $M$, $\lambda$ and $\varepsilon$. The number of eigenfactor $K$, in our paper, is set to 5 in all experiments. For dense constraint matrix $AF_{t}$, selecting larger $K$ can make the problem very computational challenging, hence not recommended. But, if the factor coefficients assigned on asset level are sparse and consequently the constraint matrix $AF_{t}$ is also sparse, the proposed method can still learn from large number of factors efficiently.  
	
	\subsubsection{Factor Risk Tolerance Learning Algorithm}
	 The learning algorithm in factor space is illustrated in Algorithm \ref{algorithm: OLF}. In sector space each mutual fund has different sector-based holdings, so $Q_{s,t}, \mathbf{c}_{s,t}$ care calculated separately for each mutual fund. In factor space, at each observation point $t$, $Q_{a,t}, c_{a,t}, F_{t}, \Sigma_{t}$ are the same for all mutual funds, and $\mathbf{c}_{f,t}$ and $\mathbf{y}_{f,t}$ are different inputs from funds.
	 
	\begin{algorithm}
		\caption{Risk Tolerance Learning from Factor Portfolio}
		\label{algorithm: OLF}
		\textbf{Input:} (time-series portfolio and price data) $ Y_{t}, P_{t}$\\
		\textbf{Initialization:} $r_{0}$ (guess), $\lambda$, $M$, $\varepsilon$ (hyper-para)    
		\begin{algorithmic}[1]
			\For{$t=1$ to $T$}  
			\State receive ($Y_{t}, P_{t}$)   
			\State $(Q_{a,t}, \mathbf{c}_{a,t}) \gets \mathbf{p}_{a,t}$   
			\State $(F_{t}, \Sigma_{t}) \gets Q_{a,t}$  \Comment eigendecomposition
			\State $(\mathbf{c}_{f,t}) \gets F_{t}^{T}\mathbf{c}_{a,t}$ 
			\State $(\mathbf{y}_{f,t}) \gets F_{t}^{T}\mathbf{y}_{a,t}$ 
			\State $\eta_{t} \gets \lambda t^{-1/2}$  \Comment{get updated $\eta_{t}$ by decaying factor}
			\State Solve $r_{f,t}$ as equation \ref{full formulation of ipo: factor space} in Section \ref{section:iop for factor space}
			\State $r_{f,t+1} \gets r_{f,t}$ \Comment{update guess of $r_{f}$}
			\EndFor
		\end{algorithmic}
		\textbf{Output:} Estimated $r_{f}$ 
	\end{algorithm}

\begin{table*}    
\begin{tabular}{l|l|l}
\hline 
& Sector Space & Factor Space \\
\hline
\textbf{Forward Problem} &  $	\begin{aligned}
	\label{forward problem of po}
	\begin{array}{llll}
	     \min\limits_{\mathbf{x}} &  \frac{1}{2} \mathbf{x}_{s}^{T} Q_{s} \mathbf{x}_{s} - r_{s} \mathbf{c}_{s}^{T} \mathbf{x}_{s}  \\
         \; s.t.  &  A_{s} \mathbf{x}_{s} \geq  \mathbf{b}_{s},
	\end{array}
	\end{aligned}$ & 
 $\begin{aligned}
	\label{forward problem of po in factor space table}
	\begin{array}{llll}
	     \min\limits_{\mathbf{x}_{f}} & \  \frac{1}{2} (\mathbf{x}_f^T\Sigma \mathbf{x}_f + \mathbf{x}_a^T Q_{d} \mathbf{x}_a)  - r_{f} \mathbf{c}_{f}^T \mathbf{x}_f\\
         \; s.t. \quad & AF \mathbf{x}_f\geq \mathbf{b}_f, 
	\end{array}
	\end{aligned}	$
	\\
\hline	
\textbf{Inverse problem} & $	\begin{aligned}
	\begin{array}{llll}
	    \min\limits_{{r}_{s}, \mathbf{x}_{s}} &\frac{1}{2}\|{r}_{s}-r_{s,t}\|^2 + \frac{\lambda}{\sqrt{t}} \| \mathbf{y}_{s,t} - \mathbf{x}_{s}\|^2\\
    	\; s.t. \quad & A_{s} \mathbf{x}_{s}\geq \mathbf{b}_{s}, \\
    	\quad & \mathbf{u} \leq M\mathbf{z}, \\
    	\quad & A_{s} \mathbf{x}_{s} - \mathbf{b}_{s} \leq M(1-\mathbf{z}), \\
    	\quad & Q_{s,t}\mathbf{x}_{s} - r_{s} \mathbf{c}_{s,t} - A_{s}^T \mathbf{u} = 0, \\
    	\quad & \mathbf{x}_{s} \in \mathbb{R}^{m} , \mathbf{u} \in \mathbb{R}_{+}^{m}, \mathbf{z} \in \{0, 1\}^{m}. 
	\end{array}
	\end{aligned}$	
	&
$ \begin{aligned}
	\begin{array}{llll}
	     \min\limits_{r_{f}, \mathbf{x}_{f}} &\frac{1}{2}\|{r}_{f} - r_{f,t} \|^2 +  \frac{\lambda}{\sqrt{t}} \| \mathbf{y}_{f,t} - \mathbf{x}_{f}\|^2\\
    	\; s.t. \quad & AF_{t} \mathbf{x}\geq \mathbf{b}_{f},  \\
    	\quad & \mathbf{u} \leq M\mathbf{z},  \\
    	\quad & AF_{t} \mathbf{x}_{f} - \mathbf{b}_{f} \leq M(1-\mathbf{z}), \\
    	\quad & \Sigma_t \mathbf{x}_{f} - r_f \mathbf{c}_{f,t} - (AF_{t})^T\mathbf{u} = 0,  \\
    	\quad & \mathbf{c}_{f,t} \in \mathbb{R}^{K}, \mathbf{x}_{f} \in \mathbb{R}^{K} , \mathbf{u} \in \mathbb{R}_{+}^{K}, \mathbf{z} \in \{0, 1\}^{K}.    
	\end{array}
	\end{aligned}	$ \\
	\hline
\textbf{Constraints} & 
$\begin{aligned}
	\label{math: A and b}
	A_{s} = \begin{bmatrix}
    I \\
    -I \\
	\mathbf{1}   \\
	-\mathbf{1}   \\
	\end{bmatrix},  \ \ & 
	\mathbf{b}_{s} = \begin{bmatrix}
    \mathbf{0} \\
    -\mathbf{1} \\
	1 \\
	-1 
	\end{bmatrix}.
	\end{aligned}$
&       $\begin{aligned}
    	AF_{t} = \begin{bmatrix}
    	F_{t}^{T} \\
    	-F_{t}^{T} \\
    	\sum_{j=1}^{n} f_{j,t}  \\
    	-\sum_{j=1}^{n} f_{j,t}   \\
    	\end{bmatrix},  \ \ & 
    	\mathbf{b}_{f} = \begin{bmatrix}
    	\varepsilon{\mathbf{1}} \\
    	-{\mathbf{1}} \\
    	1 \\
    	-1 
    	\end{bmatrix}.
    	\end{aligned}$ \\
    	\hline
\textbf{Hyper-parameters} &  $\lambda, M$ 	&  $\lambda, M, \varepsilon$ \\ 
\hline
\end{tabular}
\caption{Problem settings of risk preference learning in Sector and Factor Space.}
\label{Problem settings of risk preference learning in Sector and Factor Space.}
\end{table*}

\subsection{IPO in Sector Space and Factor Space}

\noindent In Table 2 we highlight the formulations of forward problem and inverse problem for IPO in sector space and factor space.

\subsection{Equivalence of Mean-Variance Portfolio Problem Formulations}\label{equivalence}	In mean-variance theory, the risk-tolerance parameter and the expected portfolio return parameter are closely connected and both can be used independently to trace out the mean-variance efficient frontier. Indeed, in the unconstrained scenario, the solution to the mean-variance problem 
\begin{equation}\label{version1}
\min\limits_{\mathbf{x} \in \mathbb{R}^{n}}  \frac{1}{2} \mathbf{x}^{T} Q \mathbf{x} - r \mathbf{c}^{T} \mathbf{x}
\end{equation}
can be derived by matrix differentiation, i.e., the optimal allocation weights satisfy
$$Q\mathbf{x}-r\mathbf{c}=0,$$
which in turns gives $\mathbf{x}^*=rQ^{-1}\mathbf{c}$, assuming that $Q$ is invertible. It hence follows that the targeted return $z$ satisfies
\begin{equation}\label{connection}
z=\mathbf{c}^T\mathbf{x} = r\mathbf{c}^TQ^{-1}\mathbf{c}.
\end{equation}
The solution $\mathbf{x}^*$ to (\ref{version1}) can be actually recovered by solving the following equivalent mean-variance problem 
\begin{align}\label{version2}
    \begin{array}{llll}
         \min\limits_{\mathbf{x} \in \mathbb{R}^{n}} &  \frac{1}{2} \mathbf{x}^{T} Q \mathbf{x}  \\
    	\;s.t. & \mathbf{c}^{T}\mathbf{x}=z.
    \end{array}
\end{align}
It is obvious from (\ref{connection}) that there is a one-to-one correspondence between risk aversion / tolerance parameter and the targeted return level $z$. Moreover, equation (\ref{connection}) exhibits a positive relationship between risk tolerance and targeted return level. Such a positive correlation has been observed for other variants of the mean-variance problem with constraints or in the multi-period setting (see, e.g., \cite{active}, \cite{multiperiod}). 

\subsection{Order of 11 Sectors and Example Portfolios}
\noindent We show the order of 11 sectors and the aggregated portfolios of all mutual funds at March 2020 in Table 3.

	\begin{table*}
		\centering
		\begin{tabular}{c|c|c|c|c|c|c|c}
			\hline
			Sector Name & VFINX & SWPPX & FDCAX & VHCAX & VSEQX & VSTCX & LMPRX \\
			\hline
			Basic Materials & 2.0585  & 2.0696  & 1.8866 & 0.00017 & 2.7600 & 3.3926 & 0.5647 \\
			Communication Services & 10.6418 & 10.6255 & 9.6793 & 3.7649 & 3.3698 & 2.1419 & 25.7736 \\
			Consumer Cyclical & 9.5455  & 9.4568 & 11.0045 & 8.5593 & 11.1654 & 8.2437 & 0.2692 \\
			Consumer Defensive & 8.0814  & 8.0692 & 4.6294 & 0.0489 & 3.2521 & 3.5085 & - \\
			Energy & 2.6315  & 2.6254 & 1.5087 & 1.3174 & 1.7545 & 1.9985 & 0.6166\\
			Financial Services & 13.8400 & 13.8892 & 9.6072 & 7.2355 & 13.2327 & 15.3046 & 0.9835 \\
			Healthcare & 15.1290 & 15.0931 & 14.9797 & 33.7680 & 14.0414 & 14.0102 & 32.0470\\
			Industrials   & 8.0863  & 8.0675  & 2.1276 & 10.4407 & 14.7876 & 14.7677 & 4.6743 \\
			Real Estate  & 2.9809  & 2.9807 & 3.1544 & 0 & 8.0989 & 8.0817 & - \\
			Technology & 21.4863 & 21.4337 & 27.1432  & 27.9279 & 18.2084 & 13.3433 & 29.4258 \\
			Utilities  & 3.5461  & 3.5375 & 0.9650 & 0 & 4.8055 & 3.8035 & - \\
			\hline
		\end{tabular}
		\caption{Sector level holdings (percentage) of Mutual Funds reported by March 2020}
		\label{tab:Sector level holdings (percentage) of Mutual Funds}
	\end{table*}
			
 \subsection{Computational Efficiency}
 
Comparison of computational time solving a single iteration of \ref{full formulation of ipo} or \ref{full formulation of ipo: factor space}. Each CPU time value shown in the table is the average number of 10 iterations consist of randomly sampled matrices. The performance is evaluated on a workstation equipped with AMD Ryzen Threadripper 1950X 16-Core and 64G memory. Notice that in reality we don't have 500 sectors, but the complexity of solving a 500 sector problem is the same as solving a formulation of 500 asset problem, so we mainly compare asset/sector based formulation vs. factor based formulation.
 
	\begin{table*}[th]
	\centering
	\begin{tabular}{c|c|c}
		\hline
		Formulation & Problem Size & CPU-time(seconds) \\
		\hline
		\multirow{ 5}{*}{Sector IPO} & 100 sectors &  0.466 \\
		 & 200 sectors & 5.908 \\
		 & 300 sectors & 32.396 \\
		 & 400 sectors & 46.380 \\
		 & 500 sectors & 79.485 \\
		 \hline
		\multirow{12}{*}{Factor IPO} & 400 assets 5 factors  & 3.524 \\
		& 400 assets 6 factors & 23.429 \\
		& 400 assets 7 factors & 166.250 \\
		& 400 assets 8 factors & 283.050 \\
		& 400 assets 9 factors & 293.150 \\
		& 400 assets 10 factors & 297.826 \\
        & 500 assets 5 factors  & 7.365 \\
		& 500 assets 6 factors & 40.524 \\
		& 500 assets 7 factors & 221.100 \\
		& 500 assets 8 factors & 228.550 \\
		& 500 assets 9 factors & 300.350 \\
		& 500 assets 10 factors & 300.672 \\
		\hline
	\end{tabular}
	\caption{CPU times of IPO algorithms with different sizes}
	\label{tab:CPU time compare}
	\end{table*}	    	

\clearpage
\newpage

\end{document}